\documentclass{article}
\usepackage{aaai}

\usepackage{tikz}
\usepackage{pgflibraryshapes}  
\usetikzlibrary{positioning}  
\usetikzlibrary{patterns,arrows,shapes}
\usepackage{mathrsfs}
\usepackage{paralist}
\usepackage{tabularx}
\usepackage{nicefrac}

\usepackage[centertags,fleqn]{amsmath}
\usepackage{amssymb,amsfonts}
\usepackage{txfonts}

\usepackage{array,xspace,multirow,epic}

\usepackage{theorem}

\newtheorem{lemma}{Lemma}
\newtheorem{proposition}{Proposition}

\newcommand{\qed}{\unskip\hspace*{1em}\hspace{\fill}$\Box$}
\newenvironment{proof}[1][Proof]{\begin{trivlist}
  \item[\hskip \labelsep {\it #1:}]}{%
    \qed\end{trivlist}}

			\usepackage{enumitem}
			\setenumerate[1]{label=\rm(\it{\roman{*}}\rm),ref=({\it\roman{*}}),leftmargin=*}
			\newlength{\wordlength}
			\newcommand{\wordbox}[3][c]{\settowidth{\wordlength}{#3}\makebox[\wordlength][#1]{#2}}

\usepackage[boxed]{algorithm}
\usepackage[noend]{algorithmic}
\usepackage{eqparbox}

\renewcommand{\algorithmicrequire}{\textbf{Input:}}
\renewcommand{\algorithmicensure}{\textbf{Output:}}
\algsetup{linenodelimiter=\,}
\algsetup{linenosize=\tiny}
\algsetup{indent=2em}

\usepackage{booktabs}  
\usepackage{verbatim,ifthen}
\usepackage{booktabs}  
\usepackage{enumitem}

\newcommand\eat[1]{}

\usepackage{array,xspace,multirow,epic}

	\newcommand{\pref}{\succsim\xspace}
	\newcommand{\Pref}[1][]{
		\ifthenelse{\equal{#1}{}}{\mathrel \succsim}{\mathop{R_{#1}}}
	}    
				                                      
	\newcommand{\sPref}[1][]{                  
		\ifthenelse{\equal{#1}{}}{\mathrel \succ}{\mathop{P_{#1}}}
	}                                          
	\newcommand{\Indiff}[1][]{                 
		\ifthenelse{\equal{#1}{}}{\mathrel \sim}{\mathop{\sim_{#1}}}
	}
	\newcommand{\prefset}[1][]{\ifthenelse{\equal{#1}{}}{\mathcal{R}}{\mathcal{R}_{#1}}}

\newcommand{\midd}{\mathbin{:}}

	\newcommand{\nbh}[1][]{
		\ifthenelse{\equal{#1}{}}{\nu}{\nu(#1)}
	}

	\newcommand{\cstr}[1][]{
		\ifthenelse{\equal{#1}{}}{\pi}{\cstr(#1)}
	}

			\usepackage{boxedminipage}

		\newcommand{\citet}[1]{\citeauthor{#1}~\shortcite{#1}}
		\newcommand{\citep}{\cite}
		
\newcommand{\w}{e\xspace}

\usepackage[small]{caption}
\begin{document}

	\title{The Temporary Exchange Problem}



\author{%
   Haris Aziz \and Edward Lee\\
    UNSW Sydney and Data61\\
    Sydney, Australia\\
       haris.aziz@unsw.edu.au, e.lee@unsw.edu.au
}

\maketitle

\begin{abstract}
We formalize an allocation model under ordinal preferences that is more general than the well-studied Shapley-Scarf housing market. In our model, the agents do not just care which house or resource they get but also care about who gets their own resource. This assumption is especially important when considering temporary exchanges in which each resource is eventually returned to the owner. We show that several positive axiomatic and computational results that hold for housing markets do not extend to the more general setting. We then identify natural restrictions on the preferences of agents for which several positive results do hold. One of our central results is a general class of algorithms that return any allocation that is individually rational and Pareto optimal with respect to the responsive set extension. 
\end{abstract}

\section{Introduction}

The Shapley-Scarf housing market is a well-studied formal model for barter markets where the goods can be dormitory rooms or kidneys~\citep{SoUn10a}. In the market, each agent owns a single good referred to as a house. The goal is to redistribute the houses to the agents in the most desirable fashion. \citet{ShSc74a} showed that under strict preferences, a simple yet elegant mechanism called \emph{Gale's Top Trading Cycle (TTC)}  is polynomial-time, strategyproof and find an allocation that is Pareto optimal and core stable. Even if the preferences are not strict, the algorithm can be suitably generalized while not losing any of the properties~(see e.g., \citep{AlMo11a,AzKe12a,SeSa13a,JaMa12a}). There has also been work where agents have multi-unit demand and endowments~\citep{STSY14a,Papa07c,KQW01a}. In this paper we focus on  single-unit demands. 

In the Shapley-Scarf market, agents only have preferences over houses. This is a reasonable assumption especially when the exchange is irrevocable. However, if the exchange is temporary, and the original house of an agent will be returned to her, the agent may care as to who temporarily used her house. In order to capture this additional issue, we consider the \emph{temporary exchange problem} that is a generalisation of the Shapley-Scarf housing market. In this generalisation, an agent has preferences over outcomes that take into account both what house the agent gets and also who gets her own house. The assumption of the temporary exchange also makes sense when for example a kidney patient not only cares about getting a suitable kidney but also has preference over who should get his or her donor's kidney.
The setting also applies to \emph{reinsurance markets}  in which the identity of the insurers affects the preferences over arrangements.  
For this more general setting, we
want to study fundamental questions as follows: does a core stable allocation exist and what is the complexity of finding it? What is the complexity of finding a Pareto optimal allocation?

\paragraph{Contributions}

We formulate an exchange market setting that is more general than the well-studied Shapley-Scarf market. It models several scenarios where agents are performing a temporary exchange or they care about who gets their resource.  

We first focus on core allocations in such settings and show that the core can be empty and it is NP-hard to check whether a core stable allocation exists. We also prove that finding a Pareto optimal allocation is NP-hard and testing Pareto optimality and weak Pareto optimality is coNP-complete. 

We complement the computational hardness results by presenting succinct ILP and quadratic programming formulations for finding a Pareto optimal allocation. We then consider a weakening of Pareto optimality called Pareto optimality  with respect to the responsive set extension. For this particular concept, we propose a general class of polynomial-time algorithms that return any allocation that is individually rational and Pareto optimal  with respect to the responsive set extension. 

We also consider strategic aspects and present two key impossibility results. Firstly, there exists no core-consistent and strategyproof mechanism. Secondly, there exists no individually rational, Pareto optimal, and strategyproof mechanism.  
We then identify restrictions on the preferences in particular house-predominant and tenant-predominant  preferences under which we regain the positive axiomatic and computational results that hold for the Shapley-Scarf market.

%


    \section{Temporary Exchange Problem}
    

    An instance of \emph{Temporary Exchange Problem} is a tuple $(N,H,e,\pref)$
    
    where 
    \begin{itemize}
    	\item $N=\{1,\ldots, n\}$ is the set of agents.
	\item $H=\{h_1,\ldots, h_n\}$ is the set of houses. 
	\item Endowment function $e:N\mapsto H$ maps each agent to a house. 
	Each agent $i$ owns exactly one house $e(i)$. We will denote $\bigcup_{i\in S} e(i)$ by $e(S)$. 
	\item $\pref=(\pref_1,\ldots, \pref_n)$ is the preference profile that specifies for each agent $i\in N$, the weak order preference relation $\pref_i$ over $N\times H$.\footnote{Note that in the standard housing market, the preferences are simply over the set of houses. Our model allows for more complex preferences.}
	The symbol $\pref_i$ denotes ``prefer at least as much'', $\succ_i$ denotes ``strictly more prefer'', and $\sim_i$ denotes indifference. 
    \end{itemize}

  A feasible outcome for the setting is an allocation of the houses to the agents. 
  An allocation is a one-to-one mapping from $N$ to $H$.
 If $p$ is the allocation, we will denote by $p(i)$ as the house agent $i$ gets. We will denote by $p^{-1}(h)$ the agent who gets house $h$.
  
  Each agent cares about the combination of two things: which house she gets and who gets her own house. We will refer to this combination as the \emph{outcome} for the agent.
  
  For an agent $i$, the outcome $(e(k),j)$ represents the scenario where $i$ gets house $e(k)$ and gives house $e(i)$ to agent $j$. For an agent $i$, the outcome $(e(i),i)$ represents the situation where $i$ keeps her own house. The outcome  $(e(j),j)$ represents the situation where $i$ swaps her house with $j$.
When we write that $(e(j),k)\succ_i (h_{\ell},m)$, it means that $i$ prefers outcome  $(e(j),k)$ to $(h_{\ell},m)$. 

Therefore for any allocations $p$ and $q$, an agent compares them only from the point of view of what house she gets and who gets her house:
$p\pref_i q \iff (p(i), p^{-1}(e(i))) \pref_i (q(i), q^{-1}(e(i))).$
Note that an agent $i$ will be interested in the outcome $(h_j,k)$ only if is it more preferred by her than $(e(i),i)$.  Otherwise agent $i$ would rather not be part of the exchange. 
Note that there could be multiple allocation for which the \emph{outcome} for an agent is the same.


     \bigskip

%

     \section{Properties of allocations and mechanisms}

We consider the standard properties in market design: 
(i) \emph{Pareto optimality}: there should be no allocation in which each agent is at least as happy and at least one agent is strictly happier (ii) \emph{individual rationality (IR)}: no agent should have an incentive to leave the allocation program (iii) \emph{strategyproofness}: no agent should have an incentive to misreport her preferences; and (iv) \emph{core stability}: an allocation should be such that no set of agents can form a coalition where they just exchange among themselves to get a better outcome than the the allocation.  We define these properties as follows.  
\bigskip

	An allocation $p$ is \emph{Pareto optimal} if there exists no other allocation $q$ such that $q\pref_i p$ for all $i\in N$ and $q\succ_i p$ for some $i\in N$. An allocation $p$ is \emph{weakly Pareto optimal} if there exists no other allocation $q$ such that $q\succ_i p$ for all $i\in N$.

\bigskip

An allocation $p$ is \emph{individually rational} if $(p(i),p^{-1}(e(i)))\pref_i (i,e(i))$.

\bigskip

A coalition $S\subseteq N$ \emph{blocks} an allocation $p$ on $N$ if there exists an allocation $q$ on $S$ such that for all $i\in S$, it is the case that $q(i)\in \w(S)$ and $q(i)\succ_i p(i)$. An allocation is \emph{core stable} if it admits no blocking coalition. 

\bigskip
A allocation algorithm is \emph{strategyproof} if no agent can misreport and get a better outcome. 

%
%
%
%
%
%

	\section{Core stability}
	
	We first show that unlike the Shapley-Scarf housing market, the Temporary Exchange market may not admit a core stable allocation.

	\begin{proposition}\label{prop:emptycore}
		The core of a Temporary Exchange setting instance can be empty.
		\end{proposition}
	\begin{proof}
		Consider an instance where
		\begin{itemize}
			\item $N=\{0,1,2,3,4\}$
			\item $H=\{h_0,\ldots, h_4\}$
			\item $e(i)=h_i$ for all $i\in \{0,\ldots, 4\}$
			\item The preferences are as follows.
			\begin{align*}
				1: (h_2,2), (h_0,0), (h_1,1)\\
				2: (h_3,3), (h_1,1), (h_2,2)\\
				3: (h_4,4), (h_2,2), (h_3,3)\\
				4: (h_0,0), (h_3,3), (h_4,4)\\
				0: (h_1,1), (h_4,4), (h_0,0)
				\end{align*} 
		\end{itemize}

		Our first claim is that \emph{for an allocation to be individually rational, it must be that an agent $i$ makes no exchange or she makes a pairwise exchange with agent $(i+1 \mod 5)$ or agent $(i-1 \mod 5)$.} An agent $i$ is only interested in outcomes $(h_{(i+1 \mod 5)},(i+1 \mod 5))$ or $(h_{(i-1\mod 5)},(i-1\mod 5))$ or $(h_i,i)$. Suppose $i$ gets a house different than $h_i$.  Then she gets either $h_{(i+1\mod 5)}$ or $h_{(i-1\mod 5)}$.
Suppose that agent $i$ gets house $h_{(i+1\mod 5)}$. Then the outcome is only acceptable to $i$ only if $(i+1\mod 5)$ gets house $h_i$. Similarly, if agent $i$ gets house $h_{(i-1\mod 5)}$, then the outcome is only acceptable to $i$ only if $(i-1\mod 5)$ gets house $h_i$.

		Now consider any individually rational allocation. From the claim established it follows that the allocation involves zero or more pairwise exchanges. Since there are an odd number of agents, at least one agent will not be part of any exchange. Let such as agent be $j$. Note that $j$ is interested to make an pairwise exchange with $j-1$. Agent $j-1$ also prefers outcome $(h_j,j)$ over $(h_{(j-2\mod 5)},(j-2\mod 5))$ or $(h_{(j-1\mod 5)},(j-1\mod 5))$. Hence agents in $\{(j-1\mod 5),j\}$ form a blocking coalition. 
		\end{proof}

       Next, we show that it is NP-hard to check whether a core stable allocation exists. The proof uses the example in the proof of Proposition~\ref{prop:emptycore}.
		    
		    	\begin{proposition}
		    		Checking whether there exists a core stable allocation is NP-hard if there are indifferences in the preferences and even if each agent has at most 6 acceptable outcome pairs. 
		    		\end{proposition}
		    	\begin{proof}
				We reduce from the following NP-complete problem.

				\noindent
				\textbf{Name}: {\sc ExactCoverBy3Sets (X3C)}: \\
				\noindent
				\textbf{Instance}: A pair $(R,S)$, where $R$ is a set and $S$ is a collection of subsets of 
				$R$ such that $|R|= 3m$ for some positive integer $m$ and $|s| = 3$ for each 
				$s\in S$ and each element in $R$ appears three times in $S$. \\
				\noindent
				\textbf{Question}: Is there a sub-collection $S'\subseteq S$ that is a partition of $R$? \\

				For each integer $j\in R$, we have a corresponding  gadget instance $j$ where
		  $N^j=\{0^j,1^j,2^j,3^j,4^j\}$;
	 $H^j=\{h_0^j,\ldots, h_4^j\}$; $e(i^j)=h_i^j$ for all $i\in \{0,\ldots, 4\}$
The preferences are as follows.
		    			\begin{align*}
		    				1^j:&\quad (h_2^j,2^j), (h_0^j,0^j), (h_1^j,1^j)\\
		    				2^j:&\quad (h_3^j,3^j), (h_1^j,1^j), (h_2^j,2^j)\\
		    				3^j:&\quad (h_4^j,4^j), (h_2^j,2^j), (h_3^j,3^j)\\
		    				4^j:&\quad (h_0^j,0^j), (h_3^j,3^j), (h_4^j,4^j)\\
		    				0^j:&\quad \{(e(0^k),0^i)\midd \{i,j,k\}\in S\}, (h_1^j,1^j), (h_4^j,4^j) ,(h_0^j,0^j)
		    				\end{align*}

		
		Note that agent $0^j$ is interested to perform exchanges with other agents $0^i$ and $0^k$ from other gadgets.


		The overall allocation instance involves agents and houses from all the gadgets so that $N=\bigcup N^j$ and $H=H^j$.
		
		\bigskip
		We claim that there exists a core stable allocation if and only if we have a yes instance of X3C.
		If we have a no instance of X3C, not every $0^j$ agent can form an exchange with other zero type agents so we have core deviation within  gadget $j$ such as had in the proof of Proposition~\ref{prop:emptycore}.
		Suppose we have a yes instance of X3C. In that case there is a partition of agents in $\{0^j\}$ who all get one of their most preferred outcomes. 
\end{proof}

		     \section{Pareto optimality}

		     We now turn to the problem of finding Pareto optimal allocations. 
		     
		     \subsection{Complexity of Pareto optimality}
		     
		     Since a Pareto optimal allocation is guaranteed to exist, we focus on \emph{computing} such an allocation.

		     		    	\begin{proposition}

						Checking whether there exists an allocation that is most preferred for each agent is NP-complete if we allow indifferences in the preferences and even if each agent has at most 4 acceptable outcome pairs.\footnote{Note that the problem is trivial if each agent has a unique most preferred outcome.}
		     		    		\end{proposition}
		     		    	\begin{proof}
		     				We reduce from the following NP-complete problem. \bigskip

		     				\noindent
		     				\textbf{Name}: {\sc ExactCoverBy3Sets (X3C)}: \\
		     				\noindent
		     				\textbf{Instance}: A pair $(R,S)$, where $R$ is a set and $S$ is a collection of subsets of 
		     				$R$ such that $|R|= 3m$ for some positive integer $m$ and $|s| = 3$ for each 
		     				$s\in S$ and each element in $R$ appears three times in $S$. \\
		     				\noindent
		     				\textbf{Question}: Is there a sub-collection $S'\subseteq S$ that is a partition of $R$? \\
				
		Consider the setting instance in which 
	$N=R$; $H=\{h_0,\ldots, h_{|R|}\}$; $e(i)=h_i$ for all $i\in N$. The preferences are as follows.
		     		    			\begin{align*}
   		    				j:&\quad \{(e(k),i)\midd \{i,j,k\}\in S\}, (j,h_j)
		\end{align*}

 Then there exists an allocation in which each agent gets a most preferred allocation if and only if there is a yes instance of X3C.
		     		    		\end{proof}
		     
		     \begin{proposition}
		   
		    Finding a Pareto optimal allocation is NP-hard if there are indifferences in the preferences even if each agent has at most 4 acceptable outcome pairs. 
		   \end{proposition}
		   \begin{proof}If there exists a polynomial-time algorithm to find a Pareto optimal allocation, then it will return an allocation that is most preferred for each agent if such an allocation exists. 
			   \end{proof}
		        

	            	\begin{proposition}
	            		Checking whether a given allocation is weakly Pareto optimal is coNP-complete even if preferences are strict and even if each agent has at most 4 acceptable outcome pairs. 
	            		\end{proposition}
	            	\begin{proof}
		    Use the same proof as the previous one but consider the endowment allocation. We can also make preferences strict. 
		    \end{proof}

%
%
%
%


		
		\subsection{ILP and  Quadratic programming formulations}
		
		Although computing a Pareto optimal allocation is NP-hard, one can still write an ILP to compute a maximum utility and hence Pareto optimal allocation. Since agents only express ordinal preferences, one can suppose that each agent has utility $w_{i,j,k}$ for receiving house $h_j$ and having a visitor $k$ in her own house $h_i$. 
		
		\[ \max \sum_{i,j,k \in N} x_{i,j,k} \cdot w_{i,j,k}\]
		\[ x_{i,j,k} \in \{0,1\} \text{ for } i,j,k \in N, \text{ and } \]
		\[x_{i,k,i} = x_{i,i,k} = x_{k,i,i} = 0 \text{ for } i,k \in N, i \neq k\]
		\[\text{For fixed } i \in N, \sum_{j,k \in N} x_{i,j,k} = 1 \text{ for all } j,k \in N, \]
		\[\text{For fixed } j \in N, \sum_{i,k \in N} x_{i,j,k} = 1 \text{ for all } i,k \in N, \]
		\[\text{For fixed } k \in N, \sum_{i,j \in N} x_{i,j,k} = 1 \text{ for all } i,j \in N, \]
		
		where $w_{i,j,k}$ is a weighted value corresponding to the preference of agent $i$, receiving house $h_j$ and a visitor $k$ in her own house $h_i$. 
		
		One can also write a quadratic program that is even more compact. 

		\paragraph{Quadratic programming for maximum utility}

		\[ \max \sum_{i,j \in N} x_{i,j} \cdot w_{i,j,k} \cdot x_{j,k}\]
		\[\text{For fixed } i \in N, \sum_{j \in N} x_{i,j} = 1 \text{ for all } j \in N \]
		\[\text{For fixed } j \in N, \sum_{i \in N} x_{i,j} = 1 \text{ for all } i \in N \]


		  \subsection{Responsive extension preferences}


In certain scenarios, an  agent may have underlying preferences $\pref_i^H$ over houses and over tenants $\pref_i^N$. Her preferences  over the combinations of houses and tenants may depend naturally on their underlying preferences. 
In particular, we study the situation where the preferences are based on the \emph{responsive set extension}. We consider the {responsive set extension} that is a subset of responsive preferences that only relates allocations when one allocation is unambiguously at least as preferred as another. 
		  We say that agent $i$'s preferences $\pref_i$ over $H\times N$ are \emph{responsive} if for any $j,j\in N\setminus \{i\}$ and $h,h'\in H$,  
		  			    \[(h\pref_i^H h') \wedge (j\pref_i^N j') \iff (h,j)\pref_i^{RS} (h',j').\]

		  We say that allocation $p$ is RS-PO (Pareto optimal with respect to the responsive set extension) if there exists no other allocation $q$ such that $q(i)\succsim_i^{RS} p(i)$ for all $i\in N$ and $q(i)\succ_i^{RS} p(i)$ for some $i\in N$. Note that if an allocation is not RS-PO, it admits an unambiguous improvement for the agents. 
  
		We say that an allocation $p$ is RS-IR (individually rational with respect to the responsive set extension) if $p(i)\succsim_i^{RS} (h_i,i)$. Our main result in this section is that there exists a polynomial-time algorithm to compute an RS-IR and RS-PO allocation.

	

					\begin{algorithm}[h!]
					  \caption{}
					  \label{PRA-RS}
					\renewcommand{\algorithmicrequire}{\wordbox[l]{\textbf{Input}:}{\textbf{Output}:}} 
					 \renewcommand{\algorithmicensure}{\wordbox[l]{\textbf{Output}:}{\textbf{Output}:}}
					\begin{algorithmic}
						\small
						\REQUIRE Temporary Exchange Problem $(N,H,\pref=(\pref_1,\ldots,\pref_n))$ where $\pref_i$ is a responsive preference relations composed of $\pref_i^N$ and $\pref_i^H$.
						\ENSURE RS-PO and RS-IR allocation.
					\end{algorithmic}
					\algsetup{linenodelimiter=\,}
					  \begin{algorithmic}[1] 

		\STATE Set ${\succsim_i^H}'$ and ${\succsim_i^N}'$ to indifference among all acceptable (with respect to ${\succsim_i^H}$ and ${\succsim_i^N}$) outcomes for all $i\in N$.

					\WHILE{${\succsim_i^K}'$ is not \textbf{saturated} for some $i\in N$ for some $K\in \{N,H\}$}\label{step:while-condition}
					\STATE 
					Choose such an $i$ and choose to modify ${\succsim_i^K}'$ by setting to ${\succsim_i^K}''$ to ${\succsim_i^K}'$. ${\succsim_i^K}''$ are dichotomous preferences in which only those outcomes are acceptable that are acceptable wrt to ${\succsim_i^K}'$ except the last equivalence class wrt ${\succsim_i^K}$.
					\IF{\text{RS-AA}$(N,H,(\pref_{-i}',\pref_i''))$} \label{step:IF}
					\STATE $\pref_i'\longleftarrow \pref_i''$
					\ELSE
					\STATE Else label ${\pref_i^K}'$ as \textbf{saturated}.
					\ENDIF
					\ENDWHILE

					  \RETURN \text{RS-AA}$(N,H,\pref)$.
					 \end{algorithmic}
					\end{algorithm}

		  \begin{proposition}
		  	There exists a polynomial-time algorithm that returns an individually rational and Pareto optimal allocation with respect to responsive set extension. 
		  	\end{proposition}

		  	First we outline a new problem called RS-AA. In this problem an agent gets an acceptable allocation if it gets some tenant agent that is one of the acceptable agents and a house that is one of the acceptable houses. 
	
		  	\begin{lemma}\label{lemma:acceptable}
		  		RS-AA can be solved in polynomial time.
		  		\end{lemma}
		  		\begin{proof}
		  			In the algorithm we first force symmetry in multiple ways. If an agent $i$ does not have $j$ as one of his acceptable tenants, then $h_i$ cannot be one of $j$'s acceptable houses and is removed from $i$'s list \textit{if} $i$ does not have house $h_j$ as one of this acceptable houses, then $i$ is removed from $j$'s acceptable set of agents. 
		  			After making the preferences symmetric in this way, we can simply check whether there is a perfect matching that matches each agent to one his acceptable houses. We argue why this sufficient. 


		  			We will frame the problem in a perfect matching context. Suppose there is a RS-AA problem with a set of agents $N$ and a set of houses $H$. Every agent $i$ has two most preferred sets, $H_i \subseteq H$, the preferred houses for the agent, and $N_i \subseteq N$, the preferred tenants to their own house. A RS-AA solution requires that an agent $i$ is matched to a house $h_j$ owned by agent $j$ iff $j$ finds $i$ acceptable as a tenant and agent $i$ considers house $h_j$ acceptable. Hence the setup of a matching problem over the bipartite graph with vertices $N \uplus H$ with edges $(i, h_j) \in E$ existing iff agent $i$ prefers house $h_j$ and agent $j$ prefers tenant $i$. Then a perfect matching on $(N \uplus H, E)$ corresponds one-to-one with a solution of RS-AA.
		  			 			\end{proof}

					Lemma~\ref{lemma:acceptable} is interesting because similar problems are NP-complete for several matching settings in which a match has three dimensions~\citep{CHL+16a,HHZB17a,NgHi91a}. 
					We now show that the polynomial-time algorithm for RS-AA can be used as a sub-routine to compute an RS-IR and RS-PO allocation.
		  			In order to do so we present an adaptation and generalisation of the Preference Refinement Algorithm~\citep{ABH11c} that is defined for hedonic coalition formation and that requires complete preferences. For our problem we do not have complete preferences but each preference relation has multidimensional components---one involving a house and one involving a tenant agent. 
			
		  			\begin{lemma}
		  				If RS-AA runs in polynomial time for any allocation problem, then a RS-Pareto-Optimal allocation can be computed in polynomial time. 
		  				\end{lemma}
		  				\begin{proof}
							An oracle to solve RS-AA can be used to compute an RS-Pareto-Optimal and RS-IR allocation. The details are specified in Algorithm~\ref{PRA-RS}.

		We first start with the agent's actual preferences. An acceptable allocation indeed exists: one in which each agent stays where she is. We now start changing the agents' preferences, one component preference of one agent at a time. For any given agent $i$, we see whether an acceptable allocation still exists if the least preferred acceptable elements of the component preference are now unacceptable. If an allocation exists, in that case those elements are permanently marked as unacceptable and the algorithm proceeds. 
		  			We do so until no agent's preference can be modified.  At this point, we know that a perfect allocation is Pareto optimal with respect to the responsive preferences. 
		  					\end{proof}

	\section{Incentives}
	
Up till now we have assumed that agents act sincerely and report their  truthful preferences. In this section, we explore scenarios where agents may misreport their preferences.
 We show that for the Temporary Exchange markets, strategyproofness is incompatible with other desirable axioms.
	
	\begin{proposition}
		There exists no individually rational, Pareto optimal, and strategyproof mechanism. 
		\end{proposition}	
		\begin{proof}
			Consider the following instance of the problem. 
			\begin{itemize}
				\item $N=\{0,1,2,3\}$
				\item $H=\{h_0,\ldots, h_3\}$
				\item $e(i)=h_i$ for all $i\in \{0,\ldots, 3\}$
				\item The preferences are as follows.
	\begin{align*}
					1:&\quad (h_2,2), {(h_2,0)}, (h_0,2), (h_0,0), (h_1,1)\\
					2:&\quad  {(h_3,1)}, (h_3,3),(h_1,1), (h_1,3),  (h_2,2)\\
					3:&\quad {(h_0,2}),(h_0,0), (h_2,2), (h_2,0),  (h_3,3)\\
					0:&\quad {(h_1,3)},(h_1,1), (h_3,3), (h_3,1), (h_0,0)
					\end{align*}
					\end{itemize}
					

	%
	%
	%

          	The only two individually rational allocations are  $p$ and $q$.
	\bigskip
			
			Allocation $p$:
			\begin{align*}
				1:&\quad (h_2,2), \mathbf{(h_2,0)}, (h_0,2), (h_0,0), (h_1,1)\\
				2:&\quad  \mathbf{(h_3,1)}, (h_3,3),(h_1,1), (h_1,3),  (h_2,2)\\
				3:&\quad \mathbf{(h_0,2}),(h_0,0), (h_2,2), (h_2,0),  (h_3,3)\\
				0:&\quad \mathbf{(h_1,3)},(h_1,1), (h_3,3), (h_3,1), (h_0,0)
				\end{align*}

				Allocation $q$:
		\begin{align*}
						1:&\quad \mathbf{(h_2,2)}, {(h_2,0)}, (h_0,2), (h_0,0), (h_1,1)\\
						2:&\quad  {(h_3,1)}, (h_3,3),\mathbf{(h_1,1)}, (h_1,3),  (h_2,2)\\
						3:&\quad {(h_0,2}),\mathbf{(h_0,0)}, (h_2,2), (h_2,0),  (h_3,3)\\
						0:&\quad {(h_1,3)},(h_1,1), \mathbf{(h_3,3)}, (h_3,1), (h_0,0)
						\end{align*}

	If the outcome under the original preferences is $q$, then agent 2 can misreport as follows to obtain an outcome in bold which is the only Pareto optimal and individually rational allocation under the misreported preferences.  
	\begin{align*}
		1:&\quad {(h_2,2)}, {(h_2,0)}, (h_0,2), \mathbf{(h_0,0)}, (h_1,1)\\
		2:&\quad  \mathbf{(h_3,3)},  (h_2,2)\\
		3:&\quad {(h_0,2}),{(h_0,0)}, \mathbf{(h_2,2)}, (h_2,0),  (h_3,3)\\
		0:&\quad {(h_1,3)},\mathbf{(h_1,1)}, {(h_3,3)}, (h_3,1), (h_0,0)
		\end{align*}
	
	In an individually rational allocation
	either $2$ gets outcome $(h_2,2)$ or $(h_3,3)$. If $2$ gets outcome $(h_2,2)$, then $1$ gets outcome $(h_0,0)$ or $(h_1,1)$. If $1$ gets outcome $(h_1,1)$, then the outcome is $e$ which is not Pareto optimal.  If $1$ gets outcome $(h_0,0)$, then $0$ gets outcome $(h_1,1)$. Then $3$ gets outcome $(3,h_3)$ which implies that $2$ and $3$ can swap each others' houses to get a Pareto improvement.  Thus in an individually rational allocation, $2$ gets outcome $(h_3,3)$. Hence $3$ gets outcome $(h_2,2)$. By Pareto optimality, $1$ gets outcome $(h_0,0)$ and $0$ gets $(h_1,1)$. Under the allocation in bold, $2$ gets outcome $(h_3,3)$ thereby violating strategyproofness.
	
	\bigskip
We now suppose that the outcome of the original preferences is $p$. If the outcome is $p$, then 1 can misreport to obtain an outcomes $r$ or $s$ in bold which are the only Pareto optimal and individually rational allocations under the following preferences. 	
	\bigskip
						
Allocation $r$
\begin{align*}
				1:&\quad \mathbf{(h_2,2)}, (h_1,1)\\
				2:&\quad  {(h_3,1)}, (h_3,3),\mathbf{(h_1,1)}, (h_1,3),  (h_2,2)\\
				3:&\quad {(h_0,2}),\mathbf{(h_0,0)}, (h_2,2), (h_2,0),  (h_3,3)\\
				0:&\quad {(h_1,3)},(h_1,1), \mathbf{(h_3,3)}, (h_3,1), (h_0,0)
				\end{align*}

		Allocation $s$		
	\begin{align*}
				1:&\quad {(h_2,2)}, \mathbf{(h_1,1)}\\
				2:&\quad  {(h_3,1)}, \mathbf{(h_3,3)},{(h_1,1)}, (h_1,3),  (h_2,2)\\
				3:&\quad {(h_0,2}),{(h_0,0)}, \mathbf{(h_2,2)}, (h_2,0),  (h_3,3)\\
				0:&\quad {(h_1,3)},(h_1,1), {(h_3,3)}, (h_3,1), \mathbf{(h_0,0)}
				\end{align*}

		Allocation $r$ already provides a beneficial outcome for agent $1$ thereby violating strategyproofness. Suppose the outcome is allocation $s$. In that case agent $3$ can misreport as follows to obtain the only Pareto optimal and individually rational allocation which provides an improvement for agent $3$ thereby violating strategyproofness. 
		\begin{align*}
			1:&\quad \mathbf{(h_2,2)}, (h_1,1)\\
			2:&\quad  {(h_3,1)}, (h_3,3),\mathbf{(h_1,1)}, (h_1,3),  (h_2,2)\\
			3:&\quad \mathbf{(h_0,0)},  (h_3,3)\\
			0:&\quad {(h_1,3)},(h_1,1), \mathbf{(h_3,3)}, (h_3,1), (h_0,0)
			\end{align*}

%

		This concludes the proof.	
					
			\end{proof}  
			
			We say that a mechanism is core-consistent if it returns a core stable allocation whenever a core stable allocation exists. The same argument can be used to prove that there exists no mechanism that is core-consistent and strategyproof. 
			
			\begin{proposition}
				 There exists no core-consistent and strategyproof mechanism. 
				\end{proposition}    
		        
  %
  %
  %
  %
  %
  
    In the subsequent sections, we show negative results from the previous sections can be circumvented if we consider special structure on the preferences.

  \section{Predominant Preferences}
  
  
  In this section, we consider preference restrictions under which we obtain positive axiomatic results.

		    \subsection{House-predominant Preferences}
		    
		    A particularly restricted version of consistent preferences is in which agents have strict underlying preferences over the houses,  care predominantly about the houses and use the preferences over agents as a tie-breaker. We will refer to these preferences as \emph{house-predominant} preferences.  The core is non-empty for these preferences. This follows from the fact that the TTC mechanism designed for basic housing markets works for our setting. We first describe the TTC mechanisms.

		    \begin{quote}
		    TTC: For a housing market with strict preferences, we first construct the corresponding directed graph $G(\pref)=(V,E)$ where $V=N\cup H$ and $E$ is specified as follows: each house points to its owner and each agent points to the most preferred house in the graph. Then, we start from an agent and walk arbitrarily along the edges until a cycle is completed. 
		    This cycle is removed from $G(\pref)$. Within the removed cycle, each agent gets the house he was pointing to. The graph $G(\pref)$ is \emph{adjusted} so that the remaining agents point to the most preferred houses among the remaining houses. The same step is repeated until the graph is empty. 
		    \end{quote}
		    
		    \begin{proposition}
			    For house-predominant  preferences in which agents have strict preferences over houses, the TTC rule is core stable and Pareto optimal.
			    \end{proposition}
		    \begin{proof}
			 We first argue for core stability.    When any agent $i$ is removed from the graph along with his allocated house $h$, then $h$ is a maximally preferred house for $i$ from among the remaining houses. Therefore $i$ cannot be in a blocking coalition with the agents remaining in the graph.

		We now argue for Pareto optimality. 	    
Let $S_k$ be the $k$th trading cycle that is removed from the trading cycle graph. 
In any allocation $x$ in which none of the agents are worse off than in the allocation produced by TTC, these agents must be allocated to houses in $S_1$. Taking this as the base case, it follows by easy induction that in $x$, the agents of $S_k$ must be allocated to houses in the $k$th trading cycle. Next, suppose that $i$ is a agent in $S_k$ for some $k$. Then no house in $S_k$ is more preferred by $i$ than the house that the TTC mechanism assigns him to. It follows that no agent is strictly better off in $x$ than in the allocation produced by TTC. \end{proof}

		    		    \begin{proposition}
		    			     For house-predominant preferences in which agents have strict preferences over houses, the TTC rule is strategyproof.
		    			    \end{proposition}
					    The argument is similar to the strategyproof argument for TTC for the standard Shapley-Scarf market. An agent cannot get a better house by misreporting. Suppose the agent gets the same house but a different agent takes her house. But this is not possible by the specification of TTC because the agent gets the house she points to in a cycle. 
					    

  \subsection{Tenant-predominant Preferences}

We can also turn around the problem completely so as to consider a setting in which agents have strict underlying preferences over other agents,  care predominantly about the tenants and use the preferences over houses as a tie-breaker. We will refer to such preferences as ``tenant-predominant preferences.''


		    
		    \begin{quote}
		    TTTC: For a housing market with strict preferences, we first construct the corresponding directed graph $G(\pref)=(V,E)$ where $V=N\cup H$ and $E$ is specified as follows: each agent points to its house and each house points to the agent most preferred by its owner. Then, we start from an agent and walk arbitrarily along the edges until a cycle is completed. 
		    This cycle is removed from $G(\pref)$. Within the removed cycle, each house it taken by the agent it was pointing to. The graph $G(\pref)$ is \emph{adjusted} so that the remaining houses point to the agent most preferred by its owner. The same step is repeated until the graph is empty. 
		    \end{quote}

The following propositions can be proved for TTTC and tenant-predominant preferences. The arguments for the propositions above are very similar to those of TTC for house-predominant preferences. 
	
\begin{proposition}
	For tenant-predominant preferences which agents have strict preferences over other agents, the TTTC mechanism is core stable and Pareto optimal.
	\end{proposition}

			    \begin{proposition}
			    			    	For tenant-predominant preferences which agents have strict preferences over other agents, the TTTC rule is strategyproof.
			    			    \end{proposition}

						    The same results also hold for ``tenant-only preferences'' in which agents have strict preferences over tenants and do not care about which house they get.

\section{Conclusions}

We considered a natural generalization of the Shapley-Scarf housing market with ordinal preferences. Several positive axiomatic and computational results that hold for the Shapley-Scarf housing market no longer hold for the temporary exchange market. On the other hand, we present some positive algorithmic and axiomatic results when preferences have more structure. 

The problem can be extended in several ways. Typically when exchanging holiday homes, it may be the case that the exact duration of holidays may not coincide for the people in the market. An extended model would allow for time windows and having back to back bookings. 

We presented an algorithm to compute an RS-PO allocation. One can also consider the core. A coalition $S\subseteq N$ \emph{RS-blocks} an allocation $p$ on $N$ if there exists an allocation $q$ on $S$ such that for all $i\in S$, it is the case that $q(i)\in \w(S)$ and $q(i)\succ_i^{RS} p(i)$. An allocation is \emph{RS core stable} if it admits no RS blocking coalition. 
If all preferences over houses are strict, then TTC returns an RS core stable outcome and if all preferences over agents are strict, then TTTC returns an RS core stable outcome.
For weak preferences, 
the following problems appear to be interesting. Does an RS core stable allocation always exist? What is the complexity of computing such an allocation?

         \bibliographystyle{named}
         



\begin{thebibliography}{}

         \bibitem[\protect\citeauthoryear{Alcalde-Unzu and Molis}{2011}]{AlMo11a}
         J.~Alcalde-Unzu and E.~Molis.
         \newblock Exchange of indivisible goods and indifferences: The top trading
           absorbing sets mechanisms.
         \newblock {\em Games and Economic Behavior}, 73(1):1--16, 2011.

         \bibitem[\protect\citeauthoryear{Aziz and de Keijzer}{2012}]{AzKe12a}
         H.~Aziz and B.~de~Keijzer.
         \newblock Housing markets with indifferences: a tale of two mechanisms.
         \newblock In {\em Proceedings of the 26th AAAI Conference on Artificial
           Intelligence (AAAI)}, pages 1249--1255, 2012.

         \bibitem[\protect\citeauthoryear{Aziz \bgroup \em et al.\egroup
           }{2013}]{ABH11c}
         H.~Aziz, F.~Brandt, and P.~Harrenstein.
         \newblock Pareto optimality in coalition formation.
         \newblock {\em Games and Economic Behavior}, 82:562--581, 2013.

         \bibitem[\protect\citeauthoryear{Chan \bgroup \em et al.\egroup
           }{2016}]{CHL+16a}
         P.~H. Chan, X.~Huang, Z.~Liu, C.~Zhang, and S.~Zhang.
         \newblock Assignment and pricing in roommate market.
         \newblock In {\em Proceedings of the 30th AAAI Conference on Artificial
           Intelligence (AAAI)}, pages 446--452. AAAI Press, 2016.

         \bibitem[\protect\citeauthoryear{Huzhang \bgroup \em et al.\egroup
           }{2017}]{HHZB17a}
         G.~Huzhang, X.~Huang, S.~Zhang, and X.~Bei.
         \newblock Online roommate allocation problem.
         \newblock In {\em Proceedings of the Twenty-Sixth International Joint
           Conference on Artificial Intelligence, {IJCAI-17}}, pages 235--241, 2017.

         \bibitem[\protect\citeauthoryear{Jaramillo and Manjunath}{2012}]{JaMa12a}
         P.~Jaramillo and V.~Manjunath.
         \newblock The difference indifference makes in strategy-proof allocation of
           objects.
         \newblock {\em Journal of Economic Theory}, 147(5):1913--1946, September 2012.

         \bibitem[\protect\citeauthoryear{Konishi \bgroup \em et al.\egroup
           }{2001}]{KQW01a}
         H.~Konishi, T.~Quint, and J.~Wako.
         \newblock On the {S}hapley-{S}carf economy: the case of multiple types of
           indivisible goods.
         \newblock {\em Journal of Mathematical Economics}, 35(1):1--15, 2001.

         \bibitem[\protect\citeauthoryear{Ng and Hirschberg}{1991}]{NgHi91a}
         C.~Ng and D.~S. Hirschberg.
         \newblock Three-dimensional stable matching problems.
         \newblock {\em SIAM Journal on Discrete Mathematics}, 4(2):245--252, 1991.

         \bibitem[\protect\citeauthoryear{Papai}{2007}]{Papa07c}
         S.~Papai.
         \newblock Exchange in a general market with indivisible goods.
         \newblock {\em Journal of Economic Theory}, 132:208--235, 2007.

         \bibitem[\protect\citeauthoryear{Saban and Sethuraman}{2013}]{SeSa13a}
         D.~Saban and J.~Sethuraman.
         \newblock House allocation with indifferences: a generalization and a unified
           view.
         \newblock In {\em Proceedings of the 14th ACM Conference on Electronic Commerce
           (ACM-EC)}, pages 803--820. ACM Press, 2013.

         \bibitem[\protect\citeauthoryear{Shapley and Scarf}{1974}]{ShSc74a}
         L.~S. Shapley and H.~Scarf.
         \newblock On cores and indivisibility.
         \newblock {\em Journal of Mathematical Economics}, 1(1):23--37, 1974.

         \bibitem[\protect\citeauthoryear{S{\"o}nmez and {\"U}nver}{2011}]{SoUn10a}
         T.~S{\"o}nmez and M.~U. {\"U}nver.
         \newblock Matching, allocation, and exchange of discrete resources.
         \newblock In J.~Benhabib, M.~O. Jackson, and A.~Bisin, editors, {\em Handbook
           of Social Economics}, volume~1, chapter~17, pages 781--852. Elsevier, 2011.

         \bibitem[\protect\citeauthoryear{Sonoda \bgroup \em et al.\egroup
           }{2014}]{STSY14a}
         A.~Sonoda, T.~Todo, H.~Sun, and M.~Yokoo.
         \newblock Two case studies for trading multiple indivisible goods with
           indifferences.
         \newblock In {\em Proceedings of the 28th AAAI Conference on Artificial
           Intelligence (AAAI)}, pages 791--797. AAAI Press, 2014.

         \end{thebibliography}

\end{document}